\documentclass[english]{article}

\usepackage[utf8]{inputenc}
\usepackage[T1]{fontenc}
\usepackage{babel}
\usepackage[margin=1.5cm]{geometry}
\usepackage{authblk}
\usepackage{amsmath}
\usepackage{amssymb}
\usepackage{graphicx}
\usepackage{enumitem}
\usepackage{booktabs}
\usepackage{color}
\usepackage{listings}
\usepackage{xcolor}

\lstdefinelanguage{Julia}{
  keywords={function, end, if, else, elseif, while, for, return, break, continue, struct, mutable, module, using, import, export, try, catch, finally, true, false, read_from_file, set_optimizer, optimize},
  sensitive=true,
  comment=[l]\#,
  morecomment=[n]{\#=}{=\#},
  morestring=[b]",
}

\usepackage{physics}

\usepackage{hyperref}
\usepackage{amsthm}
\usepackage[capitalise,compress]{cleveref}\crefname{section}{Sec.}{Secs.}
\Crefname{section}{Section}{Sections}
\crefrangelabelformat{equation}{\textup{(#3#1#4)}--\textup{(#5#2#6)}}
\newtheorem{thm}{Theorem}
\newtheorem{lem}{Lemma}

\crefname{cor}{Cor.}{Cors.}

\newcommand{\Prob}{\mathbb P}
\newcommand{\E}{\mathbb E}

\usepackage[style=numeric, backend=biber]{biblatex}
\usepackage{csquotes}
\title{A Critical Comment on 'Entropy Computing: A Paradigm for Optimization in Open Photonic Systems'}
\author[1]{Ali Hamed Moosavian}
\author[2]{Bahram Abedi Ravan}
\affil[1]{Ferdowsi University of Mashhad}
\affil[2]{Shahid Sattari University}
\date{} % No date to match the original

\begin{document}

\maketitle

\begin{abstract}
	In this article, we take a close look at Entropy Quantum Computing (EQC), a computational paradigm developed by Quantum Computing Inc. (QCi), which deviates from mainstream quantum computing by embracing rather than battling environmental noise and decoherence \cite{nguyen2024entropy}. In their words this approach purports EQC as an open quantum system that turns "entropy into super-power fuels of its computing engine" \cite{qci:dirac3}. We show that some of the claims in the main article can be made more rigorous, and yet these are still not good enough to beat state of the art classical algorithms on conventional classical computers. Note that these conclusions reflect the technology's current early stage of development and are not meant to discourage its pursuit. Continued rigorous exploration is necessary to fully assess the long-term viability and potential advantages of this distinct computational approach.
\end{abstract}
\section{Introduction}
Nguyen et.~al.~have introduced a novel paradigm in optimization known as entropy computing, which they propose can revolutionize solving NP-hard problems through open photonic systems. Their approach leverages photon qudits encoded in time-frequency degrees of freedom within a hybrid photonic-electronic system. The system employs feedback loops to adjust losses based on measurements, aiming to exploit quantum fluctuations or shot noise for enhanced exploration of solution spaces \cite{nguyen2024entropy}.

While their work presents an intriguing concept with potential applications, there are several aspects that warrant critical examination. First, the current implementation relies on weak coherent states rather than true Fock states, which may limit the quantum advantages they assert. This distinction is crucial as weak coherent states behave more classically, often even called classical states in the field of quantum optics, potentially affecting the system's performance and scalability (For instance see \cite[\pno~48]{leonhardtEssentialQuantumOptics2010a} or \cite[section 2.4.6 and the following sections]{waritasavanantOpticalQuantumComputers2023}). If the nature of a system is classical enough, it is expected that other classical computers should be able to simulate its evolution.

Moreover, the practicality and scope of the proposed approach remain unclear. Although the authors report results for instances involving up to 949 variables, this scale alone does not establish computational relevance for real-world NP-hard problems, where both problem size and structure play a decisive role. In particular, it is not evident how the reported performance extrapolates to larger or structurally different instances, nor whether the demonstrated success on quadratic programming and small Max-Cut problems reflects a general capability rather than favorable problem selection. These concerns are compounded by the limited diversity of benchmarks and the absence of comparisons against modern classical heuristics that are routinely used for large-scale non-convex and combinatorial optimization.

Additionally, the comparisons made with traditional methods such as gradient descent and semi-definite programming need closer inspection. It is essential to evaluate whether these benchmarks are comprehensive or if alternative approaches could offer more robust comparisons. Furthermore, the energy efficiency claims should be contextualized within the broader landscape of optical computing systems to ascertain their practical sustainability.

The remainder of this article is organized as follows. In \cref{sec:dirac3}, we summarize the operational principles of the Dirac-3 system and the class of optimization problems it is claimed to address. In \cref{sec:classical_repro}, we reproduce the main numerical benchmarks reported in \cite{nguyen2024entropy} using standard classical solvers and metaheuristic algorithms, enabling a direct comparison on identical problem instances. The statistical and thermodynamical behavior of the Dirac-3 device, viewed as an open photonic system sampling from an effective Gibbs distribution, is analyzed in \cref{sec:thermo}. Our main theoretical result, concerning the asymptotic distribution of high-value cut configurations in dense random graphs, is presented in \cref{sec:main_result} and proved in \cref{subsec:proof}. In \cref{sec:gw_comparison}, we compare the observed performance of entropy computing with established approximation thresholds and inapproximability results for MaxCut on random graphs. We conclude in \cref{sec:conclusion} with a summary of implications for claims of quantum or photonic computational advantage.

\section{Operational Principles of the Dirac-3 System}\label{sec:dirac3}

The current flagship hardware realization of EQC is the \textbf{Dirac-3 system}, described as a hybrid analog machine combining quantum optics with digital electronics. Its core operational mechanism involves an optical-feedback loop containing an optical amplifier, a photon-mode mixer/encoder, and a loss medium that induces linear and nonlinear loss mechanisms \cite{nguyen2024entropy}. Quantum states within this system are encoded as photon numbers in a train of time-bin optical pulses \cite{nguyen2024entropy} (see \cref{fig:circuit}).

\begin{figure}[ht!]
	\centering
	\includegraphics[width=\textwidth]{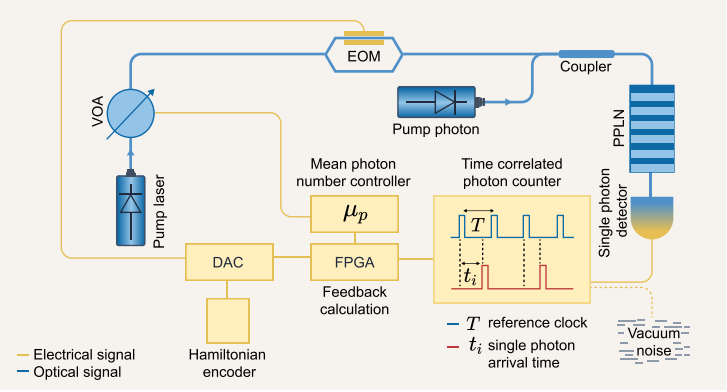} % Make sure image.png is in the same directory
	\caption{The schematics of Dirac-3. Image copied from \cite{nguyen2024entropy}, licensed under CC BY 4.0 (\href{http://creativecommons.org/licenses/by/4.0/}{http://creativecommons.org/licenses/by/4.0/}
		).}
	\label{fig:circuit}
\end{figure}

During each loop, amplified optical signals pass through a mixer, which performs linear transformations according to the target Hamiltonian \cite{nguyen2024entropy}. The QCi team claims that this process effectively simulates imaginary time evolution, where higher-energy eigenstates experience dissipation and decoherence, while lower-energy states are preferentially stabilized and emerge as solutions \cite{emami2025financial}. This means that the system iteratively relaxes toward an equilibrium point, which corresponds to the optimal solution of the given problem. A distinctive feature of the Dirac-3's implementation is its use of time-correlated single photon counting (TCSPC) and an electro-optical feedback loop to emulate the EQC paradigm.

While the system utilizes weak coherent states as input, the authors claim that maintaining a low mean photon number allows the system to approximate a single-photon regime \cite{nguyen2024entropy}. However, from a rigorous quantum optics perspective, a weak coherent state is fundamentally distinct from a single-photon Fock state. As detailed in standard texts \cite{loudon2000quantum, gerry2005introductory}, attenuating a coherent laser source reduces the mean photon number but preserves Poissonian photon statistics. Unlike a true single-photon source which exhibits sub-Poissonian statistics and anti-bunching, a weak coherent state is dominated by the vacuum component and retains a non-zero probability of multi-photon events. Consequently, the "quantum stochasticity" in this system is derived from the Poissonian shot noise of a classical field detection rather than the intrinsic non-classical properties of a pure single-photon state.

This reliance on shot noise comes with a trade-off: lower photon numbers may prolong the optimization time due to the need to overcome the noise floor \cite{nguyen2024entropy}. The nature of the paths which the Dirac-3 device explores as well as this trade-off is reminiscent of another optimization algorithm, i.e.\ Simulated Annealing (SA) \cite{Kirkpatrick2014}. In SA, similarly such a trade-off exists where increasing the temperature increases the likelihood of escaping local minima while hindering the optimization nature of the algorithm itself and therefore increasing the run-time of the algorithm.

\subsection{Computational Capabilities and Problem Solving}

EQC, as implemented in Dirac-3, is designed to tackle combinatorial optimization problems \cite{qci:dirac3}. The system supports all-to-all connectivity among variables and can handle first- through fifth-order correlations \cite{qci:dirac3, nguyen2024entropy}. The current computational mode for Dirac-3 focuses on solving continuous quadratic optimization problems, where variables are non-negative and subject to an overall sum constraint \cite{emami2025financial}. Specifically, the claim is that the machine can minimize cost functions of the form:
\begin{equation}
	\begin{split}
		E = \sum_{i} C_i v_i + \sum_{i,j} J_{ij} v_i v_j + \sum_{i,j,k} T_{ijk} v_i v_j v_k + \\
		\sum_{i,j,k,l} Q_{ijkl} v_i v_j v_k v_l + \sum_{i,j,k,l,m} P_{ijklm} v_i v_j v_k v_l v_m ,
	\end{split}
\end{equation}
where the variables and coefficients are subject to the following constraints:
\begin{itemize}
	\item The variables $v_i$ are non-negative real numbers over a discrete space, i.e., $v_i \geq 0$.
	\item The coefficients $C_i, J_{ij}, T_{ijk}, Q_{ijkl},$ and $P_{ijklm}$ are all real numbers.
	\item The tensors $J, T, Q,$ and $P$ are symmetric under any permutation of their indices (e.g., $J_{ij} = J_{ji}$, $T_{ijk} = T_{jik} = T_{kji}$, etc.).
\end{itemize}
Furthermore, a sum constraint must be applied during the optimization process such that:
\begin{equation}
	\sum_i v_i = R
\end{equation}
for some chosen constant $R$.

Practical applications demonstrated for EQC include:
\begin{itemize}
	\item \textbf{Boosting Machine Learning Techniques:} EQC has been applied to a quadratic formulation of boosting, similar to QBoost, for tasks like financial fraud detection \cite{emami2025financial}. This involves finding optimal weights for "weak" classifiers to construct a stronger classifier \cite{emami2025financial}.
	\item \textbf{Traveling Salesman Problem (TSP):} QCi's Dirac-1 system (an earlier iteration) has demonstrated the capability to handle TSP instances with up to \textbf{18 nodes} \cite{padmasola2025solving}.
	\item \textbf{Non-Convex Optimization:} One of the 3 results explained in \cite{nguyen2024entropy} is the optimization of a 2 dimensional objective function.
	\item \textbf{Continuous Quadratic Optimization:} The other result in \cite{nguyen2024entropy} is solving one of the QPLIB problems \cite{furiniQPLIBLibraryQuadratic2019} on the Dirac-3 device.
	\item \textbf{Max-Cut Problems:} While not directly involving Dirac-3, related coherent Ising machines (CIMs), which share architectural similarities with EQC, have been extensively benchmarked on MAX-CUT problems up to \emph{100,000 spins} \cite{honjo2021100000spin, ntt:performance}. In \cite{nguyen2024entropy} they compare the performance of Dirac-3 with Semi-Definite Programming for graphs with 30 vertices.
\end{itemize}
In the next section of this note, i.e.\ \cref{sec:classical_repro}, we mainly focus on the results of \cite{nguyen2024entropy}.

\section{Classically reproducing the results}
\label{sec:classical_repro}
% \colorbox{red}{\tiny should explain somewhere why it's appropriate to use metaheuristics here. The reason is that EQC doesn't have a guarantee either}

To assess the practical computational advantage claimed by the Entropy Quantum Computing (EQC) framework, we attempted to reproduce the optimization results reported for the Dirac-3 system using standard classical hardware. The benchmarks presented in the original article include a two-variable non-convex polynomial, a continuous quadratic programming instance (QPLIB\_0018), and combinatorial graph partitioning (Max-k-Cut).

We utilized a standard consumer-grade personal computer (CPU-based) for these experiments, i.e.\ 13th Gen Intel® Core™ i9-13900K × 32. The classical solvers were implemented in Julia using the \texttt{JuMP} modeling language for mathematical optimization \cite{lubinJuMP10Recent2023}, \texttt{Ipopt} for non-linear optimization \cite{wachterImplementationInteriorpointFilter2006}, \texttt{HiGHS} for exact integer programming \cite{huangfuParallelizingDualRevised2018}, and \texttt{Metaheuristics.jl} for population-based gradient-free methods \cite{metaheuristics2022}. Simulated Annealing \cite{Kirkpatrick2014} and Tabu Search \cite{Glover1989Tabu1}, two of the oldest metaheuristic algorithms, were implemented natively in Julia.

It is important to add that these comparisons with metaheuristic algorithms, are a more appropriate base of comparison than the classical algorithms that were considered in \cite{nguyen2024entropy}. The reason is that neither Dirac-3 nor the metaheuristic algorithms, unlike for example SDP, guarantee finding a good solution; but in practice we see them working very well on average.

\subsection{Two-variable non-convex optimization}

The first demonstration in the Dirac-3 report involves minimizing the polynomial function:
\begin{equation}
	g(x,y) = \frac{x^4 + y^4}{4} - \frac{5(x^3 + y^3)}{3} + 3(x^2 + y^2)
\end{equation}
subject to a box constraint. The authors compare their system against a basic Gradient Descent (GD) algorithm, noting that GD frequently becomes trapped in local minima. This comparison establishes a weak baseline, as gradient descent is known to be ill-suited for global optimization on non-convex landscapes without modification (see e.g.\ \cite[86]{goodfellowDeepLearning2016}).

We applied a standard Particle Swarm Optimization (PSO) \cite{kennedyParticleSwarmOptimization1995} and the Evolutionary Centers Algorithm (ECA) \cite{mejia-de-diosNewEvolutionaryOptimization2019} to this problem. As illustrated in \cref{fig:ECA_PSO}, the classical metaheuristics consistently converged to the global minimum $g(0,0)=0$ without being trapped in the local attractor at $(0,3)$. The convergence was achieved in negligible runtime ($<0.01$ seconds) with a population size as small as $N=14$, requiring only a few iterations to identify the global basin of attraction.

\begin{figure}[ht]
	\centering
	\includegraphics[width=0.75\textwidth]{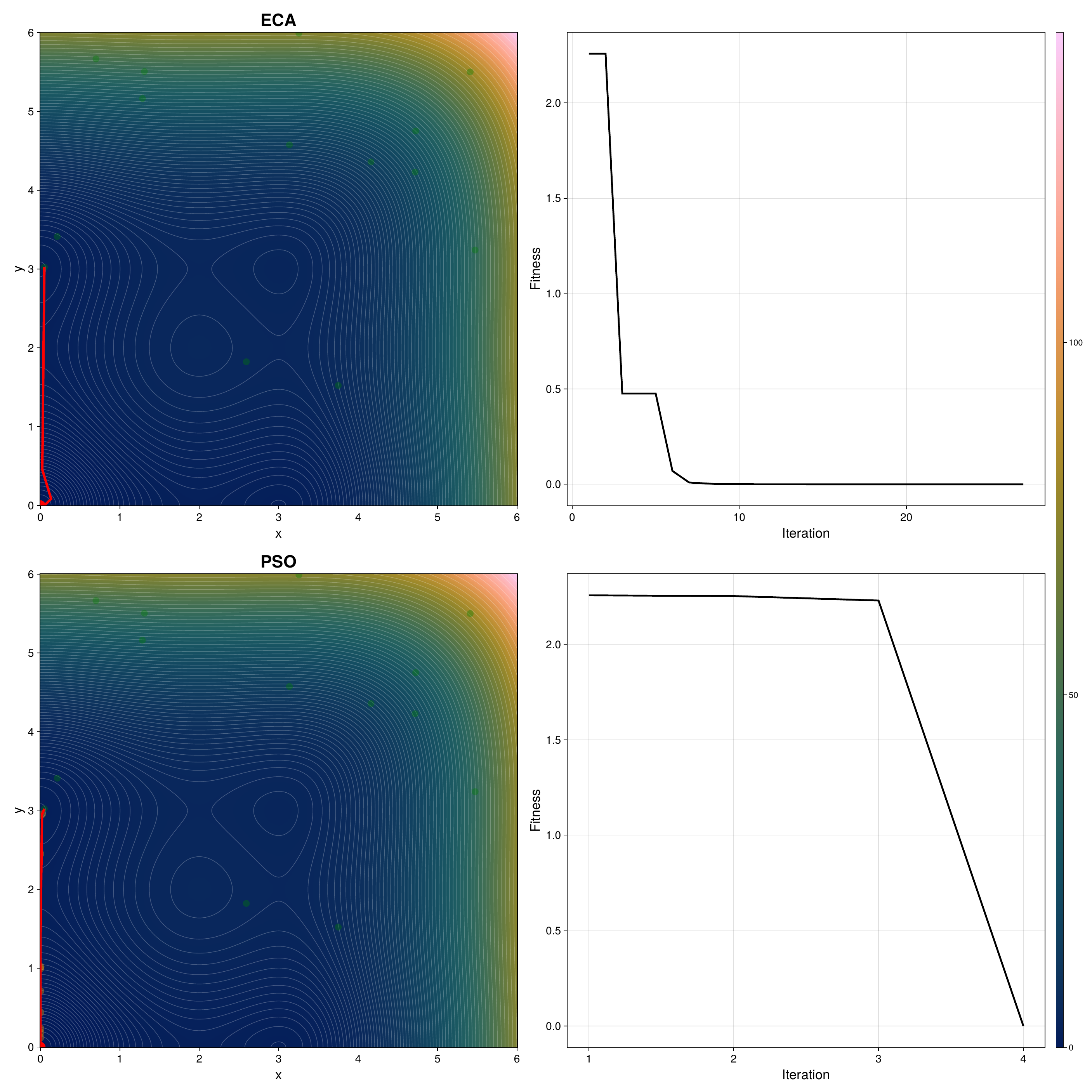}
	\caption{ECA and PSO Convergence on the non-convex polynomial g(x,y). The red line depicts the path of the optimum solution found in each iteration.}

	\label{fig:ECA_PSO}
\end{figure}
\subsection{Continuous Quadratic Optimization (QPLIB\_0018)}

The second benchmark, QPLIB\_0018, represents a non-convex quadratic optimization problem with 50 continuous variables \cite{furiniQPLIBLibraryQuadratic2019}. While the dimension is higher than the previous example, 50 variables remain well within the scope of classical non-linear solvers.

Using the interior-point optimizer \texttt{Ipopt}, we solved the QPLIB\_0018 instance to optimality. The solver converged to the known optimal solution in a fraction of a second on a single CPU core. The Dirac-3 results, while eventually finding the ground state, required tuning of the "photon number" and "quantum fluctuation coefficients" to achieve a high success probability. In contrast, the deterministic classical solver found the optimum reliably in a single run without the need for stochastic parameter tuning or analog noise management.

For reference, here's the Julia code block that solves this problem:

\begin{lstlisting}
using JuMP, Ipopt

model = read_from_file("Dirac-3/QPLIB_0018.lp")

set_optimizer(model, Ipopt.Optimizer)

@time JuMP.optimize!(model)
\end{lstlisting}

\subsection{Max-k-Cut Combinatorial Optimization}
% \colorbox{red}{this section's table and text should be updated.}

% \colorbox{red}{Improve final_benchmark_strict_fixed.pdf and add it to this section.}

The final and most complex claim involves solving \textsc{MaxCut, Max-3-Cut} \normalfont{and} \textsc{Max-4-Cut} \normalfont{on a random} graph with $n=30$ nodes. The authors compare their results against a Semi-Definite Programming (SDP) relaxation, showing that Dirac-3 often outperforms the SDP approximation.

However, for a graph of size $n=30$, SDP is an unnecessarily weak baseline. Exact solvers and simple metaheuristics can solve these instances to optimality almost instantly. We generated random graphs with $n=30$ nodes and $233$ edges. (consistent with the graph considered in \cite{nguyen2024entropy}) and solved the \textsc{MaxCut, Max-3-Cut} \normalfont{and} \textsc{Max-4-Cut} \normalfont{problems using three approaches:}
\begin{enumerate}
	\item \textbf{Exact Solver:} A Branch-and-Cut algorithm (via \texttt{HiGHS}).
	\item \textbf{Simulated Annealing (SA):} A simple native implementation.
	\item \textbf{Tabu Search:} A native implementation with a short tenure list.
\end{enumerate}

The results of our classical reproduction are summarized in \cref{tab:maxcut_bench} and \cref{fig:cuts}. The exact solver established the ground truth maximum cut values (e.g., $192$ for \textsc{Max-3-Cut} \normalfont{and} $196$ for \textsc{Max-4-Cut} \normalfont{in our generated instances).}

\begin{table}[ht]
	\centering
	\caption{Performance of classical solvers vs Dirac-3 on the Max-k-Cut ($n=30$, 100 trials). The values for the Dirac-3 device are based on the reported results in Fig.~4 of \cite{nguyen2024entropy}. "NR" means not reported. The success probabilities of EQC were calculated assuming the device was able to find the optimal cut.}
	\label{tab:maxcut_bench}
	\begin{tabular}{@{}lcccc@{}}
		\toprule[2pt]
		\textbf{Algorithm}  & \textbf{Best Cut} & \textbf{Mean Cut} & \textbf{Success Rate} & \textbf{Time (s)} \\ \midrule[2pt]
		\multicolumn{5}{c}{\textit{Max-$2$-Cut (Exact Optimum: 146)}}                                           \\
		Exact (HiGHS)       & 146               & -                 & 100\%                 & 8.7329            \\
		Simulated Annealing & 146               & 145.6             & 85.0\%                & 0.0006            \\
		Tabu Search         & 146               & 146.0             & 100.0\%               & 0.0005            \\
		\cmidrule(lr){2-5}
		EQC Sch. 1          & 146               & 141.75            & 17.0\%                & NR                \\
		EQC Sch. 4          & 146               & 141.71            & 25.0\%                & NR                \\
		\midrule[1pt]
		\multicolumn{5}{c}{\textit{Max-$3$-Cut (Exact Optimum: 192)}}                                           \\
		Exact (HiGHS)       & 192               & -                 & 100\%                 & 103.9120          \\
		Simulated Annealing & 192               & 190.6             & 14.0\%                & 0.0001            \\
		Tabu Search         & 192               & 191.2             & 28.0\%                & 0.0004            \\
		\cmidrule(lr){2-5}
		\multicolumn{5}{c}{(Optimum value $\ge$ 191)}                                                           \\
		EQC Sch. 1          & 190               & 185.76            & 0.0\%                 & NR                \\
		EQC Sch. 4          & 191               & 186.43            & $\le$5.0\%            & NR                \\
		\midrule[1pt]
		\multicolumn{5}{c}{\textit{Max-$4$-Cut (Exact Optimum: 214)}}                                           \\
		Exact (HiGHS)       & 214               & -                 & 100\%                 & 165.4461          \\
		Simulated Annealing & 214               & 212.7             & 48.0\%                & 0.0001            \\
		Tabu Search         & 214               & 213.6             & 67.0\%                & 0.0004            \\
		\cmidrule(lr){2-5}
		\multicolumn{5}{c}{(Optimum value $\ge$ 210)}                                                           \\
		EQC Sch. 1          & 210               & 206.07            & $\le$2.0\%            & NR                \\
		EQC Sch. 4          & 210               & 206.28            & $\le$4.0\%            & NR                \\
		\bottomrule[2pt]
	\end{tabular}
\end{table}

\begin{figure}[ht]
	\centering
	\includegraphics[width = 0.75\textwidth]{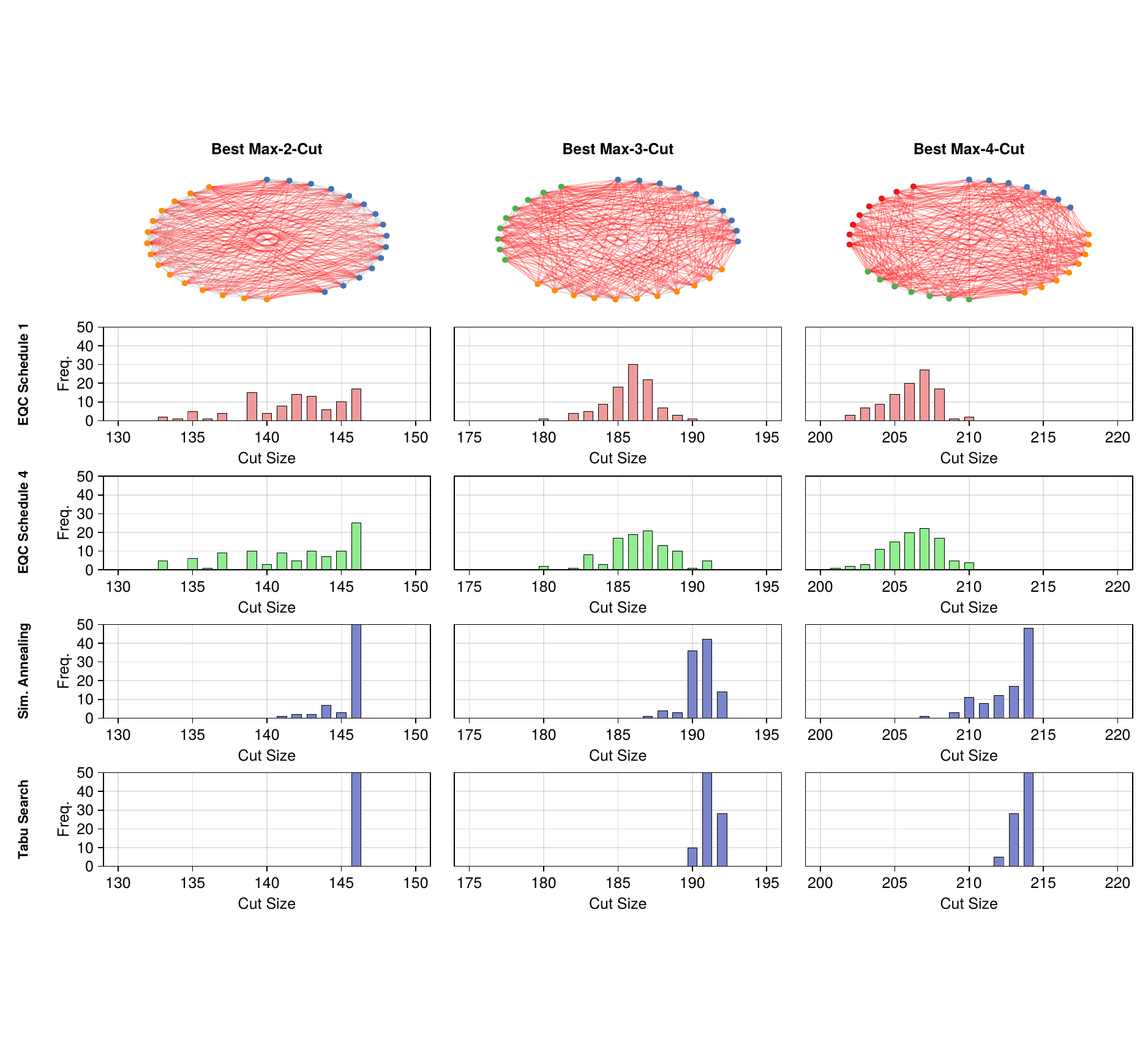}
	\caption{Comparison of calculating the \textsc{Max-k-Cut} \normalfont{problem} on the Dirac-3 machine and on a conventional classical computer. The first row represents the graph we selected for this comparison; as we didn't have access to the graph data in Fig.~4 of \cite{nguyen2024entropy} we just generated a random graph with 30 nodes and 233 edges that has the \textsc{MaxCut} \normalfont{value of 146. The next two rows are a reconstruction of the Dirac-3 performance reported in Fig.~4 of \cite{nguyen2024entropy} and do not represent the actual performance of their device on the graph in the first row. That is why instead of focusing on the cut value itself, the concentration of the distribution around the maximum cut values and the overall success probabilities should be compared. The statistical details of this figure are given in \cref{tab:maxcut_bench}. }}
	\label{fig:cuts}
\end{figure}

Crucially, the heuristic methods (Simulated Annealing and Tabu Search) found the \textit{exact} global maximum in the vast majority of trials, with mean runtimes between $0.1$ms and $0.6$ms. The Dirac-3 system, by comparison, produces a distribution of solutions that often centers around sub-optimal cuts (as seen in their Figure 4, schedules 1 and 4), only reaching the best found solutions occasionally.

The fact that a basic implementation of Simulated Annealing running in milliseconds on a consumer CPU achieves solution qualities matching or exceeding those reported for the Dirac-3 system suggests that the demonstrated problems are insufficient to establish any quantum or photonic advantage. The computational complexity of these benchmarks is simply too low to challenge modern classical algorithms.

\section{Thermodynamical behavior of states}\label{sec:thermo}
The Dirac 3 device functions as an open photonic system that exchanges energy with a thermal reservoir while maintaining a fixed number of problem variables $n$ (the number of graph nodes) and a fixed system volume. Consequently, the statistical properties of the optical modes are best described by the Canonical Ensemble \cite{pathriaStatisticalMechanics2022,nguyen2024entropy}. In this ensemble, the system relaxes toward a thermal equilibrium characterized by the Boltzmann distribution. The probability $\Prob(\sigma)$ of the system occupying a specific single cut configuration $\sigma$ with an associated energy (or cut value) $E(\sigma)$ is given by:
\begin{equation}
	\Prob(\sigma) = \frac{1}{Z} e^{-\beta E(\sigma)}
\end{equation}
where $\beta$ is the inverse effective temperature and $Z$ is the partition function (We have factored all physical constants inside $\beta$ and the energy is measured by the cut size itself). However, in combinatorial optimization tasks such as MaxCut, one observes the cut value $k$ rather than the specific microstate $\sigma$. The total probability $\mathcal{P}(k)$ of observing a cut value $k$ is determined by summing the probabilities of all distinct configurations that yield that specific cut value. If we define $\Omega(k)$ as the configuration count (or density of states) representing the number of distinct cuts resulting in value $k$, the total likelihood is expressed as:

\begin{equation} \mathcal{P}(k) = \sum_{\sigma \in \{k\}} \Prob(\sigma) = \Omega(k) \frac{e^{-\beta E_k}}{Z} \,. \label{eq:cut_dist}
\end{equation}

This relationship highlights the interplay between the energetic drive to minimize the Hamiltonian (maximize cut size) and the entropic drive to maximize the number of accessible states $\Omega(k)$. Experimental data from the Dirac 3 device utilizing a 30-node unweighted graph with 233 edges demonstrates this distribution. The device identified a maximum cut value of 146. If we had access to the graph structure tested in \cite{nguyen2024entropy} we could calculate distribution of cut configurations, $\Omega(k)$, precisely. Each of the corresponding authors of \cite{nguyen2024entropy} were contacted twice to provide the graph details, however, they did not reply. That is why we did the following to sample the graphs in Fig. 4 of \cite{nguyen2024entropy}: We adopted a reverse-engineering strategy by constructing random graphs with a planted 15-vs-15 partition of exactly 146 edges to guarantee the target cut's existence immediately. We structured internal edges to satisfy the high 3-cut and 4-cut thresholds by design, eliminating the need for blind searching. Finally, we employed batched GPU verification and massive Monte Carlo sampling with planted injection to rigorously validate the graphs and minimize statistical error. It is worth mentioning that a random graph with 30 nodes and 233 edges typically has a higher Max-Cut value, however, for the rest of our analysis we want our samples to represent the sampled graph in \cite{nguyen2024entropy} as closely as possible. We sampled 1000 graphs that met the given criteria and the results can be seen in \cref{fig:conf_count}.

\begin{figure}
	\centering
	\includegraphics[width=\textwidth]{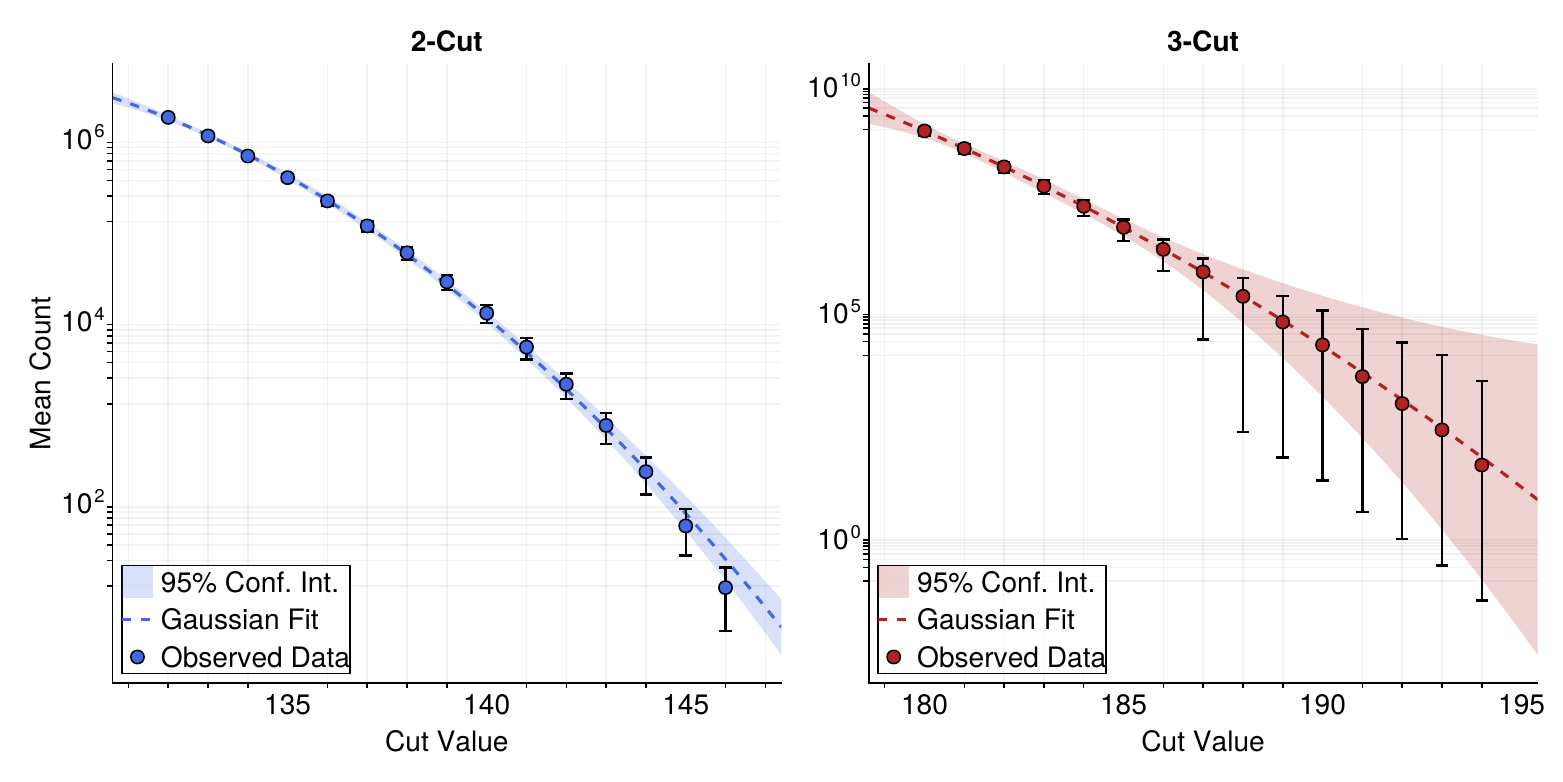}
	\caption{The number of configurations that correspond to a certain cut value, averaged over 1000 random graphs with 30 nodes, 233 edges and Max-Cut value 146, Max-3-Cut value of at least 191 and Max-4-Cut value of at least 210. The error bars are one standard deviation.}
	\label{fig:conf_count}
\end{figure}

The resulting histogram of solutions follows the distribution $\mathcal{P}(k)$, concentrating around the high-cut (low-energy) region where the Boltzmann factor $e^{-\beta E_k}$ dominates, despite the typically lower degeneracy $\Omega(k)$ of extreme values compared to the vast state space of $2^{n}$ total configurations \cite{coppersmithRandomMAXSAT2004}.

As it can be deduced from \cref{fig:conf_count}, calculating the distribution for higher cut values despite requiring heavy computations yields large error bars even with $1000$ samples. That is why we will focus on the \textsc{MaxCut} problem for the rest of this analysis.

Given the configuration counts (\Cref{fig:conf_count}), we can fit the function in \cref{eq:cut_dist} and find the effective value $\beta$ (\Cref{fig:effective_temp}). Evidently, without access to the actual structure of the examined graph in \cite{nguyen2024entropy}, these are not very tight fits. Another issue is that 100 measurements from the Dirac-3 device is not high enough to confidently display a distribution, e.g.\ is the lack of measurements in 138 cuts due to random fluctuations or structural properties of the graph? Despite these, the fits capture the essence of the distributions, the mean of the distribution in the second plot is shifted towards higher cut values, which means that it has a lower effective temperature.

\begin{figure}
	\centering
	\includegraphics[width = \textwidth]{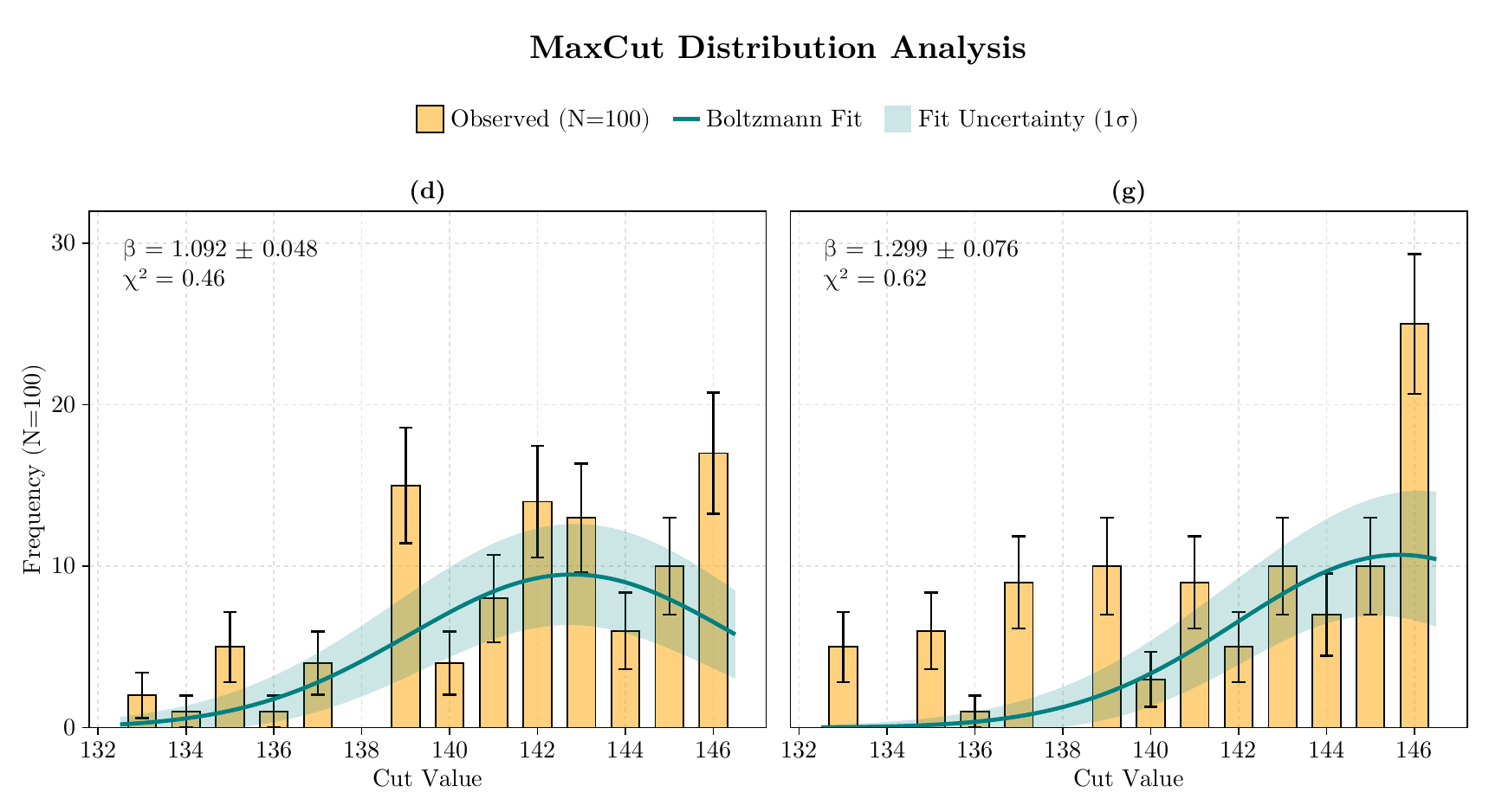}
	\caption{Model prediction versus the measurement outcomes. The fit was done by calculating the least squared value of errors in the non-linear model of \cref{eq:cut_dist}. The error bars on the observed data were calculated by $\sqrt{f(1-f/N)}$ and they represent the statistical uncertainty of the frequencies observed given the $N=100$ trials. $\chi^2$ is the reduced chi-squared statistic.}
	\label{fig:effective_temp}

\end{figure}

The statistical properties of cuts in random Erd\H{o}s-R\'enyi-Gilbert graphs $G(n,p)$ have been extensively studied, see for example \cite{friezeIntroductionRandomGraphs2016}. The maximum cut size typically concentrates around a value determined by the ground state energy of the corresponding Sherrington-Kirkpatrick spin glass model \cite{demboExtremalCutsSparse2017}. While asymptotic bounds and limiting values for the maximum cut are well-established for sparse graphs \cite{gamarnikMaxCutSparseRandom2017} and dense regimes \cite{demboExtremalCutsSparse2017}, the explicit distribution of the \textit{number} of configurations yielding specific high cut values in $G(n, 1/2)$ remains less characterized in the standard literature. It has been observed that maximum cuts exhibit near-uniqueness or cluster into specific regions of the configuration space \cite{brightwellExtremalSubgraphsRandom2009}. However, a direct distributional characterization of the enumeration of cuts at the spectral edge (the ``high cut'' regime) is often omitted in favor of studying the value of the optimum itself. Here, we utilize the first and second moment methods under the annealed approximation (analogous to the Random Energy Model) to derive the limiting Poisson distribution for the number of configurations at the high-value edge of the spectrum.

\section{Asymptotic behavior of Dirac-3 for the \textsc{MaxCut} problem}\label{sec:main_result}

We consider the Erd\H{o}s-R\'enyi-Gilbert graph $G(n, p)$ with $p=0.5$ \cite{fienbergBriefHistoryStatistical2012,erdosRandomGraphs2022,gilbertRandomGraphs1959}. Let $\mathcal{C}$ be the set of all $2^{n-1}$ possible cut configurations (partitioning vertices $V$ into two disjoint sets, identifying $(S, S^c)$ with $(S^c, S)$). Let $\Omega(k)$ be the random variable representing the number of cuts with size exactly $k$.

\begin{thm}
	In the asymptotic limit $n \to \infty$, for cut values $k$ such that the expected number of configurations $\lambda_n = \mathbb{E}[\Omega(k)]$ converges to a finite constant $\lambda > 0$, the random variable $\Omega(k)$ converges in distribution to a Poisson distribution:
	\begin{equation}
		\Omega(k) \xrightarrow{d} \text{Poisson}(\lambda).
	\end{equation}
\end{thm}

% \subsection{Mathematical Proof}

\subsection{First Moment and Regime Definition}
\begin{lem}
	Let $N = 2^{n-1}$ be the total number of cuts. For large $n$, the expected number of cuts of size $k$ is given by:
	\begin{equation}
		\mathbb{E}[\Omega(k)] \approx \frac{2^{n+1}}{n\sqrt{2\pi}} \exp\left( - \frac{8(k - n^2/8)^2}{n^2} \right).
	\end{equation}
\end{lem}

\begin{proof}
	For a fixed partition $\sigma$ dividing the set of $n$ vertices into two disjoint subsets $S$ and $S^c$, the number of potential crossing edges is given by the product of the subset sizes, $m = |S| \cdot |S^c|$. Since the existence of each edge in $G(n, 0.5)$ is determined by an independent Bernoulli trial with probability $p=0.5$, the total number of crossing edges $X_\sigma$ follows a Binomial distribution $B(m, 0.5)$. For a balanced cut, where $|S| \approx n/2$, this number of potential edges is maximized at $m \approx n^2/4$. For large $n$, we approximate this with a Gaussian distribution $\mathcal{N}(\mu, \sigma^2)$ where $\mu = n^2/8$ and $\sigma^2 = n^2/16$.
	The probability that a specific cut has value $k$ is:
	\begin{equation}
		\Prob(X_\sigma = k) \approx \frac{1}{\sqrt{2\pi \sigma^2}} e^{-\frac{(k-\mu)^2}{2\sigma^2}} = \frac{4}{n\sqrt{2\pi}} e^{-\frac{(k - n^2/8)^2}{n^2/8}}.
	\end{equation}
	Summing over all $N = 2^{n-1}$ configurations:
	\begin{equation}
		\lambda_n = \mathbb{E}[\Omega(k)] = 2^{n-1} \Prob(X_\sigma = k) \approx \frac{2^{n+1}}{n\sqrt{2\pi}} e^{-\frac{8(k - n^2/8)^2}{n^2}}.
	\end{equation}
\end{proof}
We focus on high cut values $k$ where this expectation remains finite, corresponding to the extreme tail of the distribution.

\subsection{Second Moment and Independence}\label{subsec:proof}

\begin{lem}
	For high cut values in the annealed regime, $\mathbb{E}[\Omega(k)^2] \sim \lambda_n^2 + \lambda_n$, satisfying the condition for Poisson convergence.
\end{lem}

\begin{proof}
	We expand the second moment:
	\begin{equation}
		\mathbb{E}[\Omega(k)^2] = \sum_{\sigma} \sum_{\tau} \Prob(X_\sigma = k \cap X_\tau = k).
	\end{equation}
	Separating diagonal ($\sigma = \tau$) and off-diagonal terms:
	\begin{equation}
		\mathbb{E}[\Omega(k)^2] = \sum_{\sigma} \Prob(X_\sigma = k) + \sum_{\sigma \neq \tau} \Prob(X_\sigma = k \cap X_\tau = k).
	\end{equation}
	The first term is $\mathbb{E}[\Omega(k)] = \lambda_n$. For the second term, the correlation between cut sizes $X_\sigma$ and $X_\tau$ depends on the overlap of the partitions. In the high-dimensional limit $n \to \infty$, two randomly chosen configurations are asymptotically orthogonal (overlap $\approx 0$) with probability approaching 1. For such orthogonal pairs, the cut sizes are statistically independent:
	\begin{equation}
		\Prob(X_\sigma = k \cap X_\tau = k) \approx \Prob(X_\sigma = k)\Prob(X_\tau = k).
	\end{equation}
	The contribution from highly correlated pairs (which are exponentially rare) is negligible in the annealed spectral regime. Thus:
	\begin{equation}
		\sum_{\sigma \neq \tau} \Prob(X_\sigma = k \cap X_\tau = k) \approx \sum_{\sigma \neq \tau} \Prob(X_\sigma = k)\Prob(X_\tau = k) \approx (\mathbb{E}[\Omega(k)])^2 = \lambda_n^2.
	\end{equation}
	Combining these, $\mathbb{E}[\Omega(k)^2] \to \lambda_n^2 + \lambda_n$.
\end{proof}

Since the first and second moments match those of a Poisson distribution (and higher moments can be shown to follow similarly via the Chen-Stein method or factorial moments \cite{JansonRandomGraphs2000}), the number of configurations yielding a high cut value $k$ follows:
\begin{equation}
	\Prob(\Omega(k) = m) = \frac{\lambda^m e^{-\lambda}}{m!}.
\end{equation}

\section{Comparison with Inapproximability Thresholds on \texorpdfstring{$G(n,1/2)$}{G(n,1/2)}}
\label{sec:gw_comparison}

In this section, we analyze the performance of the Dirac-3 system relative to established theoretical hardness bounds. The Goemans-Williamson (GW) algorithm provides an approximation ratio of $\alpha_{\mathrm{GW}} \approx 0.87856$ \cite{goemansImprovedApproximationAlgorithms1995}. Regarding upper bounds, H\aa stad proved that under the assumption $P \neq NP$, it is NP-hard to approximate MaxCut within a factor better than $16/17 \approx 0.941$ \cite{hastadOptimalInapproximabilityResults2001}. Furthermore, the Unique Games Conjecture (UGC) implies that the GW constant $\alpha_{\mathrm{GW}}$ is the optimal approximation ratio achievable in polynomial time for worst-case instances \cite{khotOptimalInapproximabilityResults2004,raghavendraOptimalAlgorithmsInapproximability2008}.

Despite these worst-case barriers, we demonstrate that beating these thresholds on dense random graphs such as $G(n, \tfrac{1}{2})$ is asymptotically guaranteed for Gibbs samplers, including the Dirac-3 architecture. In the remainder of this section we assume UGC to be true and use the $\alpha_{\mathrm{GW}}$ constant. However, if the conjecture turns out to be false, identical results could be obtained for the H\aa stad's inapproximability constant.

Let $G \sim G(n, \tfrac{1}{2})$ be a random graph, and let $C(\sigma)$ denote the cut size associated with a configuration $\sigma \in \{\pm 1\}^n$. The Dirac-3 device samples from the Boltzmann distribution:
\begin{equation}
	\Prob_\beta(\sigma \mid G) = \frac{e^{\beta C(\sigma)}}{Z(\beta)}, \quad \text{where} \quad Z(\beta) = \sum_{\sigma} e^{\beta C(\sigma)} = \sum_k \Omega(k) e^{\beta k}.
\end{equation}
The induced probability law on the observed cut values $C$ is therefore:
\begin{equation}
	\Prob_\beta(C=k \mid G) = \frac{\Omega(k) e^{\beta k}}{\sum_{k'} \Omega(k') e^{\beta k'}}.
\end{equation}

We define the event $\mathcal{S}$ wherein the sampled cut exceeds the GW threshold relative to the maximum cut $C_{\max}(G)$:
\begin{equation}
	\mathcal{S} := \left\{ C(\sigma) > \alpha_{\mathrm{GW}} C_{\max}(G) \right\}.
\end{equation}

\begin{thm}
	\label{thm:beats_gw}
	For a random graph $G \sim G(n, \tfrac{1}{2})$ and any fixed inverse temperature $\beta > 0$, the probability that the Dirac-3 system samples a cut value exceeding the Goemans-Williamson worst-case bound converges to 1 as $n \to \infty$. Specifically:
	\begin{equation}
		\Prob_\beta(\mathcal{S} \mid G) \ge 1 - \exp\left( - \Theta(n^2) \right).
	\end{equation}
\end{thm}

\begin{proof}
	We begin with a graph-by-graph bound. The probability of the complement event $\mathcal{S}^c$ is given by the total measure of states below the threshold normalized by the partition function:
	\begin{equation}
		\Prob_\beta(\mathcal{S}^c \mid G) = \frac{\sum_{\sigma : C(\sigma) \le \alpha_{\mathrm{GW}} C_{\max}} e^{\beta C(\sigma)}}{Z(\beta)}.
	\end{equation}
	Bounding the numerator by the total number of states $2^n$ multiplied by the maximum weight in the set, and bounding the denominator by the weight of the optimal solution alone ($e^{\beta C_{\max}}$), we obtain:
	\begin{equation}
		\Prob_\beta(\mathcal{S}^c \mid G) \le \frac{2^n e^{\beta \alpha_{\mathrm{GW}} C_{\max}(G)}}{e^{\beta C_{\max}(G)}} = \exp\left( n \ln 2 - \beta(1 - \alpha_{\mathrm{GW}}) C_{\max}(G) \right).
		\label{eq:prob_bound}
	\end{equation}
	For $G(n, \tfrac{1}{2})$, the maximum cut scales as $C_{\max}(G) = \frac{n^2}{8} + \Theta(n^{3/2})$ with high probability \cite{demboExtremalCutsSparse2017,mezardSpinGlassTheory2004}. Substituting this scaling into \cref{eq:prob_bound} and assuming $\beta \approx 1$ (consistent with the effective temperature observed in experiments):
	\begin{equation}
		\Prob(\mathcal{S}^c) \le \exp\left( n \ln 2 - \frac{1 - \alpha_{\mathrm{GW}}}{8} n^2 + o(n^2) \right) = \exp\left( - \Theta(n^2) \right).
	\end{equation}
	Thus, the system exceeds the threshold with overwhelming probability.
\end{proof}

In other words, \cref{eq:prob_bound} states that if in future versions of these devices the effective temperature grows linearly with the number of graph nodes they can handle, then they would retain the same probability for beating SDP's worst case promise on random graph instances.

\subsection{Approximation via Gaussian Analysis}

We can further refine this analysis using the annealed approximation for the density of states derived in the previous section. Let $\lambda_n(k) = \E[\Omega(k)]$. The probability of observing a cut in the range $[k_{\mathrm{GW}}, k_{\max}]$, where $k_{\mathrm{GW}} := \alpha_{\mathrm{GW}} k_{\max}$, is approximately:
\begin{equation}
	\Prob_\beta(k_{\mathrm{GW}} \le C \le k_{\max}) \approx \frac{\int_{k_{\mathrm{GW}}}^{k_{\max}} \lambda_n(k) e^{\beta k} \, dk}{\int_{-\infty}^{k_{\max}} \lambda_n(k) e^{\beta k} \, dk}.
\end{equation}
Using the Gaussian approximation for $\lambda_n(k)$ centered at $\mu = n^2/8$, the integrand takes the form $\exp\left( -a(k-m)^2 \right)$ where $a = 8/n^2$ and the shifted mean is $m = \mu + \beta n^2/16$. Defining the standard normal variable $z(x) = 4(x-m)/n$, the probability can be expressed via the standard normal CDF $\Phi$:
\begin{equation}
	\Prob_\beta(k_{\mathrm{GW}} \le C \le k_{\max}) \approx 1 - \frac{\Phi(z(k_{\mathrm{GW}}))}{\Phi(z(k_{\max}))}.
\end{equation}
Given the expansion $k_{\max} = \mu + P_* \frac{n^{3/2}}{4} + o(n^{3/2})$ where $P_* \approx 0.7632$ is the ground-state constant of the SK model \cite{panchenkoSherringtonKirkpatrickModelOverview2012}, we find that $z(k_{\mathrm{GW}})$ scales as $-n(1-\alpha_{\mathrm{GW}})/2$. Consequently, $\Phi(z(k_{\mathrm{GW}}))$ vanishes exponentially, confirming the result of \cref{thm:beats_gw}.

\subsection{Implications for Quantum Advantage Claims}

While formally correct, the result in \cref{thm:beats_gw} does not constitute evidence of quantum advantage nor a challenge to standard complexity classes. The hardness results established by H\aa stad ($16/17$) and those implied by the UGC ($\alpha_{\mathrm{GW}}$) apply to \textit{worst-case} adversarial instances. In contrast, for dense random graphs such as $G(n, \tfrac{1}{2})$, the optimal cut is $C_{\max} \approx n^2/8 + \Theta(n^{3/2})$, while a random assignment yields $n^2/8$ in expectation. The approximation ratio for even a trivial random guess is therefore:
\begin{equation}
	\frac{n^2/8 + O(n)}{\text{OPT}(G)} = \frac{n^2/8}{n^2/8 + \Theta(n^{3/2})} \to 1 \quad \text{as } n \to \infty.
\end{equation}
Since $1 > 16/17 > \alpha_{\mathrm{GW}}$, trivial classical algorithms theoretically beat both H\aa stad's bound and the GW constant on this distribution asymptotically. Furthermore, Polynomial-Time Approximation Schemes (PTAS) exist for MaxCut on graphs with density $\Omega(n^2)$ \cite{aroraPolynomialTimeApproximation1999}. Consequently, observing that Dirac-3 exceeds $\alpha_{\mathrm{GW}}$ on $G(n, \tfrac{1}{2})$ merely confirms that the system avoids catastrophic failure modes; it does not demonstrate computational superiority over standard classical heuristics in the regime where the problem is actually hard.
\section{Conclusions}\label{sec:conclusion}

In this work, we have critically examined the claims of computational advantage put forward in \cite{nguyen2024entropy} for the so-called entropy computing paradigm implemented in the Dirac-3 photonic system. Our analysis combined a reproduction of the reported numerical benchmarks, a statistical interpretation of the device dynamics, and a theoretical investigation of the optimization landscapes relevant to the studied \textsc{MaxCut} instances.

Our numerical results demonstrate that the reported performance of the Dirac-3 system can be closely matched by standard classical optimization methods and widely used metaheuristics when applied to the same problem instances. This observation suggests that the benchmark problems considered in \cite{nguyen2024entropy} do not, by themselves, provide evidence of a qualitative computational advantage, particularly in the absence of comparisons against state-of-the-art classical solvers or harder instance families.

From a physical and statistical perspective, we have shown that the observed behavior of the Dirac-3 system is consistent with sampling from an effective Gibbs distribution at a problem-dependent temperature. While such behavior is interesting in its own right as an instance of analog stochastic optimization in an open photonic system, it does not imply access to non-classical computational resources. Rather, it places the device within a broader class of classical sampling-based heuristics whose performance is governed by well-understood thermodynamical and algorithmic principles.

Our main theoretical result further clarifies the limitations of the reported \textsc{MaxCut} performance. For dense random graphs, we showed that a broad class of sampling-based algorithms will, with high probability, produce cuts with values exceeding a fixed threshold above the trivial baseline, independent of any system-specific dynamics. Consequently, the observation of near-optimal or high-value cuts on such instances does not constitute evidence of enhanced optimization capability unless it is accompanied by scaling analyses or performance guarantees that surpass known approximation thresholds.

Taken together, these findings indicate that the results reported in \cite{nguyen2024entropy} are best interpreted as a demonstration of a physically motivated heuristic optimizer rather than as evidence for a new paradigm offering computational or quantum advantage. Establishing such an advantage would require substantially larger and more diverse benchmarks, direct comparisons with modern classical algorithms, and clear demonstrations of scaling behavior that cannot be explained within existing algorithmic or statistical frameworks.

We conclude that while entropy computing in open photonic systems is an interesting direction for experimental exploration, the current evidence does not support the stronger claims regarding its optimization power or its relevance to hard combinatorial problems. Future work along this line would benefit from closer integration with the theoretical literature on approximation algorithms, complexity, and statistical mechanics of optimization, as well as from more rigorous benchmarking standards.

\section{Code Availability}
The code generated during and/or analysed during the current study is available from the corresponding author on reasonable request.
\section{Acknowledgements}
During the preparation of this work, the authors used Google Gemini to assist with the generation of some of the Julia programs. After using this tool, the authors reviewed and edited the content as needed and take full responsibility for the content of the publication.

\newpage

\printbibliography

\end{document}

\maketitle

\section{Metaheuristic algorithms getting out of local minima}
Here we show more examples of two metaheuristic algorithms, Particle Swarm Optimization (PSO) \cite{kennedyParticleSwarmOptimization1995} and Evolutionary Centers Algorithm (ECA) \cite{mejia-de-diosNewEvolutionaryOptimization2019}, finding the global minimum of the two dimensional function:
\begin{equation}
        g(x,y) = \frac{x^4 + y^4}{4} - \frac{5(x^3 + y^3)}{3} + 3(x^2 + y^2)\, .
\end{equation}

The two metaheuristic algorithms used are two population based algorithms. In this case we used $N=14$ population count and forced the initial populations, visible as dull green dots, to be centered around the three local minima.
\begin{figure}[ht]
    \centering
    \includegraphics[width = \textwidth]{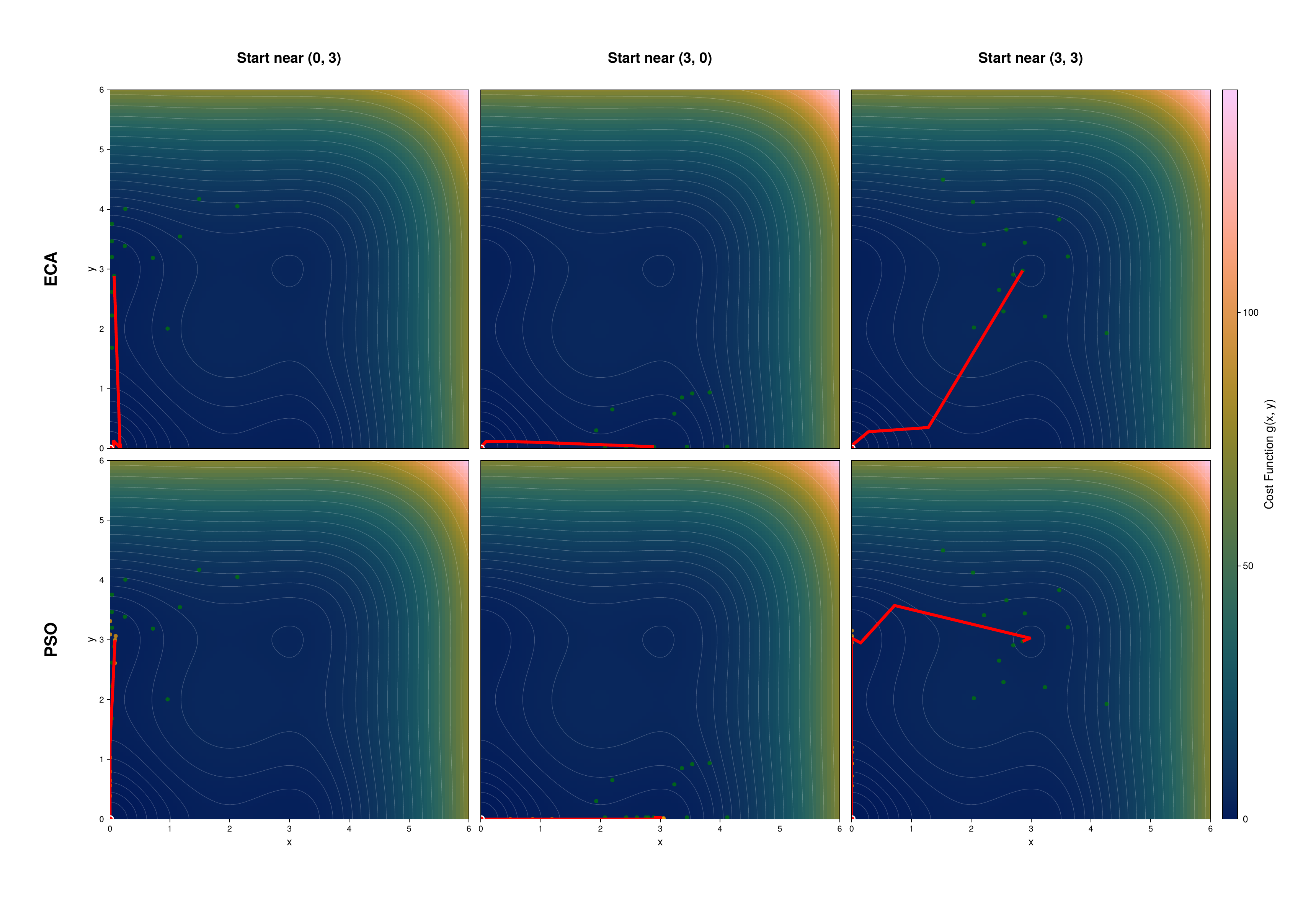}
    \caption{The plot shows how Evolutionary Centers Algorithms \cite{mejia-de-diosNewEvolutionaryOptimization2019} and Particle Swarm Optimization \cite{kennedyParticleSwarmOptimization1995} can find the global minimum of a cost function despite starting near its local minima.}
    \label{fig:ECA_PSO}
    
\end{figure}

\section{System Specifications and Operational Environment}

The Dirac-3 system is engineered with the following specifications:
\begin{itemize}
    \item \textbf{Variables:} 949 \cite{qci:dirac3}.
    \item \textbf{Connectivity:} All-to-all \cite{qci:dirac3}.
    \item \textbf{Architecture:} Qudit-based \cite{qci:dirac3}.
    \item \textbf{Power Consumption:} Under 100 Watts \cite{qci:dirac3}.
    \item \textbf{Operating Temperature:} 20°C to 27°C, operating at room temperature without the need for cryogenic or isolated housing \cite{qci:dirac3}.
    \item \textbf{Dimensions:} 5U rack-mounted \cite{qci:dirac3}.
\end{itemize}
The system's design leverages the principles of quantum mechanics in open quantum systems, which differ from conventional quantum computers that require precise isolation from environmental noise \cite{qci:dirac3}.

\section{Stochastic Nature and Solution Search Strategy}

EQC is fundamentally a \textbf{stochastic computational framework} that deliberately utilizes non-deterministic processes in its solution search strategy.
\begin{itemize}
    \item \textbf{Random Noise and Quantum Fluctuations:} The intrinsic randomness arising from quantum fluctuation of single-photon counting shot noise is continuously introduced into each feedback loop. This deliberate injection of noise is critical for the system to explore vast search spaces and escape local minima \cite{nguyen2024entropy}. The system is maintained at ultra-low mean photon numbers (less than 1\% of pulses having more than one photon) to maximize this quantum stochasticity \cite{nguyen2024entropy}.
    \item \textbf{Probabilistic Search and Dissipative Dynamics:} The mapping of a target Hamiltonian onto an effective dissipative operator guides the evolution of photon qudit states. This process introduces controlled loss and decoherence, probabilistically suppressing unwanted states and favoring the emergence of lower-energy states as solutions \cite{emami2025financial, nguyen2024entropy}. The iterative evolution of the system's state allows it to relax toward an optimal solution, and its trajectories can exhibit "quantum-tunneling-like" behavior on non-convex energy landscapes, aiding in escaping local minima \cite{nguyen2024entropy}.
    \item \textbf{Annealing-like Behavior:} The dissipative dynamics of EQC share similarities with annealing processes, where the system gradually settles into a low-energy state corresponding to the optimal solution.
    \item \textbf{Quantum Zeno Effect (QZE):} EQC leverages the Quantum Zeno effect, where frequent measurements or strong dissipation can "freeze" the evolution of a quantum system \cite{berwald:grover}. This is implemented through controllable loss mechanisms \cite{berwald:grover} and has been theoretically shown to support an \textbf{optimal Grover-like quantum speedup} ($\sqrt{N}$ scaling) for unstructured search problems \cite{berwald:grover}. The quantum Zeno blockade, involving high photon loss rates engineered through nonlinear optics, is a key component in the entropy computing paradigm \cite{berwald:grover}.
\end{itemize}

\newpage

% --- PRINT BIBLIOGRAPHY ---
\printbibliography